\documentclass[12pt,a4paper]{article}
\usepackage{graphicx}
\usepackage{amsthm}
\theoremstyle{plain}

\newtheorem*{proposition}{Proposition}
\title{Simple equations method (SEsM)  and some of its numerous particular cases}
\author{Nikolay K. Vitanov$^1$, Zlatinka I. Dimitrova$^2$}
\date{$^1$Institute of Mechanics, Bulgarian Academy of Scienecs, Acad. G. Bonchev Str., Bl. 4, 1113 Sofia, Bulgaria \\
$^2$ "G. Nadjakov" Institute of Solid State Physics, Bulgarian Academy of Sciences, Bldv. Tzarigradsko Chaussee 72, 1784 Sofia, Bulgaria}
\begin{document}
\maketitle

\begin{abstract}
We discuss a  new version of a method for obtaining exact solutions of nonlinear partial differential equations. We  call this method  the
Simple Equations Method (SEsM). The method is based on representation of the searched solution as function 
of solutions of one or several  simple equations. We show that SEsM
contains as particular case the Modified Method of Simplest Equation, G'/G - method, Exp-function method,
Tanh-method and the method of Fourier series for obtaining exact and approximate solutions of linear
differential equations. These methods are only a small part of the large amount of methods that are 
particular cases of the methodology of SEsM.
\end{abstract}

\section{Introduction}
Complex systems attract the attantion of numerous researchers today  \cite{holmesx} - \cite{cs5}. 
Many of these systems are nonlinear \cite{n1} - \cite{n8} and  the nonlinearity is investigated by 
various methodologies such as  time series analysis or by means of models based on differential equations
\cite{t1} - \cite{t10}. In numerous cases the model equations are nonlinear partial differential equations
and one is interested in obtaining exact analytical solutions of these equations. The  research on the
methodology for obtaining such exact solution was connected  several decades ago to the attempts to avoid the nonlinearity by
appropriate transformation of the solved differential equation, e.g., by the Hopf-Cole transformation \cite{hopf}, \cite{cole}  which transforms the nonlinear Burgers 
equation to the linear heat equation.  The study of  such transformations led to the
\emph{Method of Inverse Scattering Transform} \cite{ablowitz} - \cite{gardner} and to 
the direct method of Hirota for obtaining of  exact solutions of NPDEs -  \emph{Hirota method} \cite{hirota}, \cite{hirota1} which
is based on bilinearization of the solved nonlinear partial differential equation by means of appropriate transformations. Truncated Painleve expansions may lead to many of these appropriate transformations \cite{tabor} - \cite{wtk}. Following this observation  Kudryashov \cite{k3} studied the possibility
for obtaining exact solutions of NPDEs by a truncated Painleve expansion where the truncation happens after the "constant term" (i.e., the constant term is kept in the expansion).  He  formulated   the \emph{Method of Simplest Equation (MSE)} \cite{k05} based 
on determination of singularity order $n$ of the solved NPDE and searching of
particular solution of this equation as series containing powers of solutions
of a simpler equation called \emph{simplest equation}. The methodology further developed in \cite{kl08}
and applied  for obtaining traveling wave solutions of  nonlinear partial differential equations
(see, e.g., \cite{k5a} - \cite{k12}). 
\par
Here we shall discuss the methodology of  the \emph{Simple Equations Method} (SEsM). 
Some elements of this methodology can be observed in
our articles written almost three decades years ago \cite{mv1} - \cite{mv5}. 
Our research continued in  2009
\cite{1}, \cite{2} and in  2010 we have used the ordinary differential 
equation of Bernoulli as simplest equation \cite{v10} and  applied the methodology of the
Modified Method of Simplest Equation to ecology
and population dynamics \cite{vd10}. In these publications we have used the concept of the balance equation. 
Note that the version  called \emph{Modified Method of Simplest Equation - MMSE} \cite{vdk}, \cite{v11} based on determination of the kind
of the simplest equation and truncation of the series of solutions of the simplest equation by means 
of application of a balance equation is equivalent of the \emph{Method of Simplest Equation}.
One can consider the \emph{Modified Method of Simplest Equation} as a version of the \emph{Method
of Simplest Equation}. Up to now our contributions to the methodology and its application have been
connected to the \emph{MMSE} \cite{v11a} - \cite{vdv17}. We note especially the article \cite{vdv15} where we have extended the methodology 
of the \emph{MMSE} to simplest equations of the class
\begin{equation}\label{sf}
\left (\frac{d^k g}{d\xi^k} \right)^l = \sum \limits_{j=0}^{m} d_j g^j
\end{equation}
where $k=1,\dots$, $l =1,\dots$, and $m$ and $d_j$ are parameters. The solution of Eq.(\ref{sf}) defines
a special function that contains as particular cases, e.g.,: (i) trigonometric functions; (ii) hyperbolic functions; (iii) elliptic functions of Jacobi; (iv) elliptic function of Weierstrass.
\par 
In the course of time our studies have been directed to the goal to extend the methodology of the 
Modified Method of Simplest Equation. The last modification of the modified method of simplest equation 
is connected to the possibility of use of more than one simplest equation. This modification can be called   
MMSEn (Modified Method of Simplest equation based on $n$ simplest equations)
but it is better to call it SEsM - Simple Equations Method. The reason for this is that the used simple equations 
are more simple than the solved nonlinear partial differential equation but these simple equations in fact can be 
quite complicated.   A variant of SEsM based on two simplest equation was applied in \cite{vd18}
and the first description of the methodology was made in \cite{v19} and then
in \cite{v19a} and \cite{v19b}. For more applications of particular cases of the methodology see
\cite{n17}, \cite{iv19}. 
\par
The text below is organized as follows. We describe the SEsM methodology in Sect. 2.
In Sect. 3 we prove several propositions in order to show that several
frequently used methods for obtaining exact solutions of nonlinear partial differential
equations as well as the method of Fourier series for obtaining exact and approximate solutions
of linear partial differential equations are particular cases of the SEsM methodology. Several
concluding remarks are summarized in Sect. 4.
\section{The Simple Equations Method (SEsM)}

Let us consider a nonlinear partial differential equation 
\begin{equation}\label{eqx}
{\cal{DE}}(u,\dots)=0
\end{equation}
where ${\cal{DE}}(u,\dots)$ depends on the function $u(x,...,t)$
and some of its derivatives ($u$ can be a function of more than 1 spatial coordinate).
The steps of the methodology of SEsM are as follows.
\begin{description}
	\item[Step 1.]
	In order to transform the nonlinearity of the solved equation to
	more treatable kind of nonlinearity, e.g., polynomial nonlinearity or even in order to remove the nonlinearity (if possible) we perform a transformation
	\begin{equation}\label{m1}
	u(x,\dots,t)=T(F(x,\dots,t))
	\end{equation}
	where $T(F)$ is some function of another function  $F$. In general
	$F(x,\dots,t)$ is a function of the spatial variables as well as of the time. 
	No general form of this transformation is known up to now.
	The transformation $T(F)$ can be
	\begin{itemize}
		\item the Painleve expansion \cite{hirota}, \cite{k3}, 
		\item $u(x,t)=4 \tan^{-1}[F(x,t)]$ for the case of the 
		sine - Gordon equation ,
		\item  $u(x,t) = 4 \tanh^{-1}[F(x,t)]$ for the case of sh-Gordon (Poisson-Boltzmann 
		equation) (for applications of the last two transformations see, e.g. \cite{mv1} - \cite{mv5}),
		\item another transformation.
	\end{itemize}
	In many particular cases one may skip this step (then we have just $u(x,\dots,t)=F(x,\dots,t)$) 
	but in numerous cases the step is necessary
	for obtaining a solution of the studied nonlinear PDE. The application of Eq.(\ref{m1}) to 
	Eq.(\ref{eqx}) leads to a nonlinear PDE for the function $F(x,\dots,t)$.
	\item[Step 2.]
	The function $F(x,\dots,t)$ is represented as a function of other functions $f_1,\dots,f_N$.
	These functions are  connected to solutions of some differential equations which can be partial 
	or ordinary differential equations) that are more simple than Eq.(\ref{eqx}). We note that 
	the possible values of $N$ are $N=1,2,\dots$ (there may be infinite number of functions $f_i$ too).
	No general form of the function $F(f_1,\dots,f_N)$ is known up to now.
	The forms of the function $F(f_1,\dots,f_N)$ can be different, e.g.,
	\begin{itemize} 
		\item
		\begin{eqnarray}\label{m2}
		F &=& \alpha + \sum \limits_{i_1=1}^N \beta_{i_1} f_{i_1} + \sum \limits_{i_1=1}^N  \sum \limits_{i_2=1}^N 
		\gamma_{i_1,i_2} f_{i_1} f_{i_2} + \dots + \nonumber \\
		&&\sum \limits_{i_1=1}^N \dots \sum \limits_{i_N=1}^N \sigma_{i_1,\dots,i_N} f_{i_1} \dots f_{i_N}
		\end{eqnarray}
		where $\alpha,\beta_{i_1}, \gamma_{i_1,i_2}, \sigma_{i_1,\dots,i_N}\dots  $ are parameters.
		\item
		or $F(f_1,\dots,f_N)$ can have another form
	\end{itemize}
	We shall use the form from Eq.(\ref{m2}) below. Note that 
	\begin{itemize}
		\item 
		the relationship (\ref{m2}) contains as particular case the 
		relationship used by Hirota \cite{hirota}. 
		\item
		The power series $\sum \limits_{i=0}^N \mu_n f^n$ (where
		$\mu$ is a parameter) used in the previous versions of the methodology based on 1 simple equation (and called  Modified Method of Simplest Equation) are a particular case of the relationship (\ref{m2}) too.
	\end{itemize}
	\item[Step 3.] 
	In general the functions $f_1,\dots,f_N$ are solutions of partial differential equations. These equations are more simple than the
	solved nonlinear partial differential equation. There are two
	possibilities
	\begin{itemize} 
		\item 
		One may use solutions of the simple partial differential equations if such solutions are available or
		\item 
		one transforms the more simple partial differential equations 
		by means of appropriate ans{\"a}tze (e.g.,  traveling-wave ans{\"a}tze such as 
		$\xi = \hat{\alpha} x + \hat{\beta} t$;  
		$\zeta = \hat{\mu} y + \hat{\nu}t \dots$). Then 
		the solved   differential equations for $f_1,\dots,f_N$ may be reduced to   differential equations 
		$E_l$, containing derivatives of one or several functions
		\begin{equation}\label{i1}
		E_l \left[ a(\xi), a_{\xi},a_{\xi \xi},\dots, b(\zeta), b_\zeta, b_{\zeta \zeta}, \dots \right] = 0; \ \
		l=1,\dots,N
		\end{equation}
	\end{itemize}
	In numerous  cases (e.g, if the equations for the functions $f_1,\dots$ are ordinary differential equations) one may skip this step 
	but the step may be necessary if the equations for $f_1,\dots$ are complicated partial differential equations.
	\item[Step 4.]
	We assume that	the functions $a(\xi)$, $b(\zeta)$, etc.,  are  functions of 
	other functions, e.g., $v(\xi)$, $w(\zeta)$, etc., i.e.
	\begin{equation}\label{i1x}
	a(\xi) = A[v(\xi)]; \ \ b(\zeta) = B[w(\zeta)]; \dots
	\end{equation} 
	Note that the kinds of the functions $A$ , $B$, $\dots$ are not prescribed. 
	Often one uses a finite-series relationship, e.g., 
	\begin{equation}\label{i2}
	a(\xi) = \sum_{\mu_1=-\nu_1}^{\nu_2} q_{\mu_1} [v (\xi)]^{\mu_1}; \ \ \ 
	b(\zeta) = \sum_{\mu_2=-\nu_3}^{\nu_4} r_{\mu_2} [w (\zeta)]^{\mu_2}, \dots 
	\end{equation}
	where $q_{\mu_1}$, $r_{\mu_2}$, $\dots$ are coefficients.
	However other kinds of relationships may be used too. 
	\item[Step 5.]
	The functions  $v(\xi)$, $w(\zeta)$, $\dots$ 
	are solutions of simple ordinary differential equations.
	For about a decade we have used particular case of the described 
	methodology that was based on use of one simple equation. This simple equation was called simplest equation and the
	methodology based on one equation was called Modified Method of Simplest Equation. SEsM contains the Modified Method of Simplest Equation as  particular cease (as one of the numerous particular cases of the SEsM methodology).
	\item[Step 6.]
	The application of the steps 1. - 5. to Eq.(\ref{eqx}) leads to change of the left-hand side of this equation. We consider the case when the result of this change  is a function that is a sum of terms where each 
	term contains some function multiplied by a coefficient. This coefficient contains some of the 
	parameters of the solved equation and some of the parameters of the solution. In the most cases
	a balance procedure must be applied in order to ensure that the above-mentioned relationships
	for the coefficients contain more than one term ( e.g., if the result of the transformation 
	is a polynomial then the balance procedure has to ensure that the coefficient of each 
	term of the polynomial is a relationship that contains at least two terms).
	This balance procedure may lead to one or more additional relationships among the parameters 
	of the solved equation and parameters of the solution. The last relationships are called 
	\emph{balance equations}. 
	\item[Step 7.]
	We may obtain a nontrivial solution of Eq. (\ref{eqx})  if all coefficients mentioned in Step 6. are
	set to $0$. This leads in the most cases  to a system of nonlinear algebraic equations for the 
	coefficients of the solved nonlinear PDE and for the coefficients of the solution. Any nontrivial 
	solution of this algebraic system leads to a solution the studied  nonlinear partial differential 
	equation. Usually the above system of algebraic equations contains many equations that have to 
	be solved with the help of   a computer algebra system. The algebraic system however can be
	complicated enough and then it is not possible to solve it even by means of a computer algebra system. 
	In this case we can not obtain exact solution of the solved nonlinear partial differential equation
	by SEsM but ne can try to solve the algebraic system numerically and if this is successful then
	SEsM will lead to a numerical solution to the solved nonlinear partial differential equation.  
\end{description}
Below we shall show that several much used methods for obtaining exact and approximate solutions
of differential equations are particular cases of SEsM.
\section{The other methods as particular cases of the SEsM methodology}
Many of the other methodologies for obtaining exact solutions of the nonlinear partial differential equations 
are particular cases of the SEsM. In the following subsections we prove several propositions
about this. 
\subsection{$G'/G$ method as particular case of SEsM methodology}
We shall start with the $G'/G$-method. This method is as follows.
We have a nonlinear partial differential equation ${\cal{DE}}^*u(x,t)=0$ where
${\cal{DE}}^*$ is the corresponding differential operator. The ansatz $\xi =
\alpha x + \beta t$ ($\alpha$ and $\beta$ are parameters) reduces the solved differential equation to the ODE
${\cal{DE}}u(\xi)=0$ and the solution of this equation is searched as
\begin{equation}\label{c1}
u(\xi) = \sum \limits_{m=0}^{M} \alpha_m \left[\frac{G'(\xi)}{G(\xi)} \right]^{m}
\end{equation}
where $\alpha_m$ are parameters, $G'= dG/d\xi$ and the function $G(\xi$)
is a solution to the linear ordinary differential equation ($\lambda$ and $\mu$ are parameters)
\begin{equation}\label{c2}
\frac{d^2 G}{d \xi^2} + \lambda \frac{dG}{d \xi} + \mu G =0
\end{equation}
We shall denote this method $(G'/G)_1$-method. Dimitrova \cite{dim12} proposed a
generalization of this method. We shall call this generalization $(G'/G)_n$-method. 
The $(G'/G)_n$ - method is as follows. Let $G''/G$ be a polynomial of $G'/G$
\begin{equation}\label{c3}
\frac{G''}{G}= \sum \limits_{n=0}^N \beta_n \left(\frac{G"}{G} \right)^n
\end{equation}
where $\beta_n$ are parameters. Then we obtain the following chain of methods
\begin{description}
	\item[1.)] The $(G'/G)_1$-method (The convetional $G'/G$-method). Here $N=1$ and
\begin{equation}\label{c4}
\frac{G''}{G} = \beta_1 \frac{G'}{G} + \beta_0
\end{equation} 
We set here $\lambda = - \beta_1$ and $\mu=- \beta_0$ and obtain Eq.(\ref{c2}).
	\item[2.)] The $(G'/G)_2$-method . Here $N=2$ and
\begin{equation}\label{c5}
\frac{G''}{G} = \beta_2 \left( \frac{G'}{G}\right)^2 + \beta_1 \frac{G'}{G} + \beta_0
\end{equation} 	
\item[3.)] The $(G'/G)_3$-method . Here $N=3$ and
\begin{equation}\label{c6}
\frac{G''}{G} = \beta_3\left( \frac{G'}{G}\right)^3 + \beta_2 \left( \frac{G'}{G}\right)^2 + \beta_1 \frac{G'}{G} + \beta_0
\end{equation} 	
\begin{center}
	$\dots$
\end{center}
\item[n.)] The $(G'/G)_n$-method . Here $N=n$ and
\begin{equation}\label{c7}
\frac{G''}{G} = \beta_n\left( \frac{G'}{G}\right)^n + \dots + \beta_3\left( \frac{G'}{G}\right)^3 + \beta_2 \left( \frac{G'}{G}\right)^2 + \beta_1 \frac{G'}{G} + \beta_0
\end{equation} 
\end{description}
Let us now show that the chain of the $(G'/G)_N$-methods is particular case of the SEsM methodology.
\begin{proposition}
The chain of the $(G'/G)_N$-methods is particular case of the SEsM for the case of 1 simplest equation of the
kind 
$$
\frac{dw}{d\xi} = -w^2  + \sum \limits_{j=0}^N \beta_j w^j,
$$
and for the case of solution of the kind
$$
u(\xi) = \sum \limits_{m=0}^M \alpha_m w^m
$$
\end{proposition}
\begin{proof}
Let us consider the methodology of the Simple Equations Method (SEsM) for the case of 1 simple equation.
Let the simple equation be
\begin{equation}\label{c8}
\frac{dw}{d\xi} = -w^2 +  \sum \limits_{j=0}^N \beta_j w^j,
\end{equation}
and let us search for the solution of the solved nonlinear differentioal equation ${\cal{DE}}u(\xi)=0$ which
is a polynomial of the solution of the simple equation (\ref{c8}), i.e.,
\begin{equation}\label{c9}
u(\xi) = \sum \limits_{m=0}^M \alpha_m w^m
\end{equation}
Let us now set
\begin{equation}\label{c10}
\frac{G'}{G} = w
\end{equation}
Then $\frac{G''}{G}=w' + w^2$ and the simple equation (\ref{c8}) is transformed to
\begin{equation}\label{c11}
\frac{G''}{G} = \beta_N\left( \frac{G'}{G}\right)^n + \dots + \beta_3\left( \frac{G'}{G}\right)^3 + \beta_2 \left( \frac{G'}{G}\right)^2 + \beta_1 \frac{G'}{G} + \beta_0
\end{equation} 
which contains as particular cases the equations of the chain of methods $(G'/G)_1$, $(G'/G)_2$, $\dots$,
$(G'/G)_n$, $\dots$.
\end{proof}
Let us write the simple equation Eq.(\ref{c8}) for the first members of the chain of the $(G'/G)_N$=methods.
\begin{description}
	\item[1.)] $N=1$: $\frac{dw}{d \xi} = - w^2 + \beta_1 w + \beta_0$ which is a version of the Riccati
	equation. This relationship was established by Kudryashov \cite{k10}
	\item[2.)] $N=2$: $\frac{dw}{d \xi} = (\beta_2- 1) w^2 + \beta_1 w + \beta_0$ which is again a version of the Riccati equation. 
\item[3.)] $N=3$: $\frac{dw}{d \xi} = \beta_3 w^3+ (\beta_2- 1) w^2 + \beta_1 w + \beta_0$. This equation can be reduced to a version of the equation of Bernoulli ($\beta_2=1$, $\beta_0=0$). 
\item[4.)] $N=4$: $\frac{dw}{d \xi} = \beta_4 w^4 + \beta_3 w^3+ (\beta_2- 1) w^2 + \beta_1 w + \beta_0$. This equation can be reduced to a version of the equation of Bernoulli ($\beta_2=1$, $\beta_0=0$). 
\end{description}
\subsection{Exp-function method as particular case of the SEsM methodology}
Within the scope of the Exp-function method \cite{efm} the solution of the solved nonlinear equation ${\cal{DE}}u(\xi)=0$ ($\xi = \alpha x + \beta t$) is searched as
\begin{equation}\label{d1}
u(\xi) = \frac{\sum \limits_{i=0}^m a_i \exp(i \xi)}{\sum \limits_{j=0}^n b_j \exp(j \xi)}
\end{equation}
\begin{proposition}
Let $k=max(m,n)$. Then the Exp-function method is a particular case of the SEsM methodology for the
case of $k$ simple equations of the kind $\frac{df_l}{d \xi} = l f_l $, $l=0,1,\dots, k$ and
solution of the differential equation searched as a function of the kind
$$ u(\xi) = \frac{\sum \limits_{i=0}^m [a_i f_i(\xi)/\exp(\xi_{0i})]}{\sum \limits_{j=0}^n [b_j f_j(\xi)/exp(\xi_{0j})]}$$ where $\xi_{0i}$ and $\xi_{0j}$ are constants of integration.
\end{proposition}
\begin{proof}
The solution of the simple equation $\frac{df_l}{d \xi} = l f_l $, $l=0,1,\dots, k$ is $f_l = \exp(l \xi + \xi_{0l})$ where $\xi_{0l}$ is a constant of integration. The substitution of this solution in 
$$ u(\xi) = \frac{\sum \limits_{i=0}^m [a_i f_i(\xi)/\exp(\xi_{0i})]}{\sum \limits_{j=0}^n [b_j f_j(\xi)/\exp(\xi_{0j})]}$$ leads to Eq.(\ref{d1}). Thus the  Exp-function method is a particular case of the SEsM methodology.
\end{proof}
\par 
We can construct numerous extensions of the Exp-function method of the basis of the SEsM methodology.
Let us show one of them based on generalization of the function $u(\xi)$. Let us consider the SEsM methodlogy based on the
	$k=max(m,n)$ simple equations $\frac{df_l}{d \xi} = l f_l $, $l=0,1,\dots, k$ and on the function
	$$
	u(\xi) = \sum \limits_{k=0}^K A_k \left[ \frac{\sum \limits_{i=0}^m [a_i f_i(\xi)/\exp(\xi_{0i})]}{\sum \limits_{j=0}^n [b_j f_j(\xi)/\exp(\xi_{0j})]} \right]^{B_k}
	$$ where $A_k$ and $B_k$ are parameters. For $K=0$, $A_0=1$ and $B_0=1$ we obtain the conventional version of the exp-function method.
\par 
SEsM methodology allows for many other generalizations of the exp-function method. For another example of
generalization of the exp-function method see \cite{v11a}.
\subsection{Tanh-method as particular case of the SEsM methodology}
The Tanh-method \cite{malf} is as follows. One searches for solution of the nonlinear equation ${\cal{DE}}u(\xi)=0$ as follows:
\begin{description} 
	\item[1.)] Case without boundary conditions:
\begin{equation}\label{e1}
u(\xi) = \sum \limits_{i=0}^N a_i v(\xi)^i, \ \ \ v(\xi) = \tanh(\xi)
\end{equation}
\item[2.)] Case of vanishing $u(\xi)$ at $\xi \to - \infty$
\begin{equation}\label{e2}
u(\xi) = (1+v)^m\sum \limits_{i=0}^{N-m} a_i v(\xi)^i, \ \ \ v(\xi) = \tanh(\xi), \ \ m=1,\dots, N
\end{equation}
\item[3.)] Case of vanishing $u(\xi)$ at $\xi \to \pm \infty$
\begin{eqnarray}\label{e3}
u(\xi) &=& (1-v)^p(1+v)^q\sum \limits_{i=0}^{N-p-q} a_i v(\xi)^i, \ \ \ v(\xi) = \tanh(\xi), \nonumber \\
&&\ p(\ne 0) + q(\ne 0) = 2, \dots, N
\end{eqnarray}
\begin{proposition}
The Tanh-method is particular case of the SEsM methodology for the case of use of 1 simplest equation (the
differential equation of the function $\tanh(\xi)$) and for the following forms of the function $u(\xi)$:
\begin{enumerate} 
\item
$u(\xi) = \sum \limits_{i=0}^N a_i v(\xi)^i$ for the case without boundary conditions; 	
\item 
$u(\xi) = (1+v)^m\sum \limits_{i=0}^{N-m} a_i v(\xi)^i, \ \ m=1,\dots, N$ for the case of vanishing $u(\xi)$ at $\xi \to - \infty$
\item 
$u(\xi) = (1-v)^p(1+v)^q\sum \limits_{i=0}^{N-p-q} a_i v(\xi)^i, \ \ p(\ne 0) +
q(\ne 0) = 2, \dots, N$ for the case  of vanishing $u(\xi)$ at $\xi \to \pm \infty$
\end{enumerate}
\end{proposition}
\begin{proof}
Let us consider the case of 1 simple equation. Let the simple equation be $\frac{dv}{d\xi} = 1 - v^2 $.
The solution of this simple equation is $v(\xi) = \tanh(\xi)$. Let us set the function $u(\xi)$ from the
SEsM methodology as
\begin{enumerate} 
	\item
	$u(\xi) = \sum \limits_{i=0}^N a_i v(\xi)^i$ for the case without boundary conditions; 	
	\item 
	$u(\xi) = (1+v)^m\sum \limits_{i=0}^{N-m} a_i v(\xi)^i, \ \ m=1,\dots, N$ for the case of vanishing $u(\xi)$ at $\xi \to - \infty$
	\item 
	$u(\xi) = (1-v)^p(1+v)^q\sum \limits_{i=0}^{N-p-q} a_i v(\xi)^i, \ \ p(\ne 0) +
	q(\ne 0) = 2, \dots, N$ for the case  of vanishing $u(\xi)$ at $\xi \to \pm \infty$
\end{enumerate}
Then the SEsM methodology is reduced to the Tanh-method, i.e., the Tanh-method is particular case
of the SEsM methodology.
\end{proof}
\end{description}
\subsection{The Modified Method of Simplest Equation as particular case of the SEsM methodlogy}
The steps of the methodology of the Modified Method of Simplest Equation for obtaining particular traveling wave  
solutions of a NPDE reduced to ${\cal{DE}}u(\xi)=0$ are:
\begin{enumerate}
	\item
	 $u(\xi)$  is represented as a function of other function $v$
	that is solution of some ordinary differential equation (the 
	simplest equation). The form of the function $u(v)$ is can be different. One example is
	\begin{equation}\label{mx2}
	u = \sum \limits_{i=-M}^N \mu_n v^n 
	\end{equation}
	$\mu$ is a parameter. In the most cases one uses $M=0$.
	\item
	The application of Eq.(\ref{mx2}) to ${\cal{D}}u(\xi)=0$ transforms the left-hand side of 
	this equation. Let the result of this transformation  be a function that is a sum of terms where each 
	term contains some function multiplied by a coefficient. This coefficient contains some of the 
	parameters of the solved equation and some of the parameters of the solution. In the most cases
	a balance procedure must be applied in order to ensure that the above-mentioned relationships
	for the coefficients contain more than one term ( e.g., if the result of the transformation 
	is a polynomial then the balance procedure has to ensure that the coefficient of each 
	term of the polynomial is a relationship that contains at least two terms).
	This balance procedure leads to one  relationship among the parameters 
	of the solved equation and parameters of the solution. The  relationship is called 
	\emph{balance equation}. 
	\item
	We may obtain a nontrivial solution of ${\cal{DE}}u(\xi)=0$  if all coefficients mentioned above are
	set to $0$. This condition usually leads to a system of nonlinear algebraic equations for the 
	coefficients of the solved nonlinear PDE and for the coefficients of the solution. Any nontrivial 
	solution of this algebraic system leads to a solution the studied  nonlinear partial differential 
	equation. Usually the above system of algebraic equations contains many equations that have to 
	be solved with the help of   a computer algebra system. 
\end{enumerate}
\begin{proposition}
The Modified Method of Simplest Equation is particular case of the SEsM methodology for the case of
use of 1 simple equation (this one used in the corresponding realization of the Modified Method of Simplest Equation) and appropriate choice of $u(\xi)$ (this one used in the corresponding realization of the Modified Method of Simplest Equation).
\end{proposition}
\begin{proof} 
Let us consider the SEsM methodology for the case of one simple equation. Let us choose this simple equation to be the same as the simplest equation used in the realization of the Modified method of simplest
equation under consideration. Let us choose the form of the function $u(\xi)$ in the SEsM methodology
to be the same as in the realization of the Modified Method of Simplest Equation under consideration. 
Then the SEsM methodology is reduced to the corresponding realization of the Modified Method of Simplest Equation. 
The above procedure can be performed for any realization of the Modified Method of Simplest Equation. 
Thus the Modified Method of Simplest Equation is a particular case of the SEsM methodology.
\end{proof} 
The last proposition shows that numerous methods based on different single simple equations, e.g., for the elliptic function of Weierstrass, for the elliptic function of Jacobi, etc., and on solution $u(\xi)$
constructed as power series or other functions of these solutions of the corresponding simplest equation,
are particular case of the SEsM methodology for the case of 1 simple equation. 
\subsection{A very general case that is particular case of the SEsM methodology}
\begin{proposition}
	Let us consider any method of solving the equation ${\cal{DE}}u(x,t)=0$ that is based on
	solutions $v_1(\xi_1), \dots,v_n(\xi_n)$, $\xi_1 = \alpha_i x + \beta_i t + \gamma_i$ ($\alpha_i$,
	$\beta_i$,$\gamma_i$: parameters) of $n$ simple differential equations ${\cal{O}}_iu_i(\xi_i)=0$,
	$i=1,\dots,n$ and
	let the function $u(x,t)$ be searched as an arbitrary combination of these simple equations. Then
	this method is a particular case of the SEsM methodology. 
\end{proposition}
\begin{proof}
Let us cosider the SEsM methodology based on $n$ simple differential equations ${\cal{O}}_iu_i(\xi_i)=0$,
$i=1,\dots,n$ and on the function $u(x,t)$ that can be any combination of these simple equations. Then the
SEsM methodology is reduced to the corresponding method, i.e., the corresponding method is particular case of the SEsM methodology.
\end{proof}
Note that in the above proposition the differential equations  ${\cal{O}}_iu_i(\xi_i)=0$ are
ordinary differential equations. The SEsM methodology extends even to the case when the $n$
simple equations are partial differential equations and even to the case when the coefficients
that participate together with the solutions $v(x,t)$ of these equations depend on $x$ and $t$ and
even to the case of stochastic simple differentiel equations, simple differential equations with fractional derivatives, etc. 
\par 
Let us now describe just one consequence of the last proposition. Namely
\begin{proposition}
The method of Fourier series for obtaining exact and approximate solutions of linear partial differential equations is particular case of SEsM.
\end{proposition}
\begin{proof}
Fourier series are widely used for seeking solutions to  ordinary differential equations  and partial differential equations. Let us consider the differential equation ${\cal{DE}}u(\xi)=0$ ($\xi = \alpha x + \beta t$, $\alpha$ and $\beta$ - parameters) and let us search the solution of this equation as
\begin{equation}\label{fourier}
u(\xi) = \frac{a_0}{2} + \sum \limits_{k=1}^\infty [a_k v_k(\xi) + b_k w_k(\xi)]
\end{equation}
where $a_0$, $a_k$ and $b_k$ are parameters ($k=1,\dots, \infty$) and the functions $v_k$ and $w_k$
are different solutions of the simple equations of the same kind:
\begin{equation}\label{fourier1}
\frac{d^2 v_k}{d \xi^2} = -k^2 v_k; \ \ \ \frac{d^2 w_k}{d \xi^2} = -k^2 w_k 
\end{equation}
i.e. ($v_k(\xi)= \cos(k \xi)$; $w_k(\xi) = \sin(k \xi)$). Then the SEsM is reduced to the
method of Fourier series for searching of traveling wave solutions of differential equations.
In analogous way we can show that the Fourier method for searching of standing wave solutions
of differential equations is particular case of SEsM. Even more in analogous way we can show that the method of 
orthogonal functions for searching of solutions of differential equations is a particular case of SEsM.
\end{proof}

\section{Concluding remarks}
In this text we have discussed the SEsM methodology for obtaining exact and approximate solutions
of nonlinear partial differential equations.
We have shown that several methods for obtaining exact particular solutions of nonlinear partial 
differential equations are particular cases of the SEsM methodology. Many more methods are
particular cases of this methodology.  We shall report additional material on this topic elsewhere.  

\nocite{*}
\bibliographystyle{aipnum-cp}%


\end{document}